\documentclass[10pt,a4paper,conference]{IEEEtran}
\usepackage{graphicx,amsmath,latexsym,amssymb}

\begin{document}

\newcommand{\redundancy}{{\mathcal{R}}}
\newcommand{\gen}{{\cal P}}
\newcommand{\model}{{\cal M}}
\newcommand{\expdif}[1]{D^{(#1)}}
\newcommand{\partf}{\textnormal{Z}}
\newcommand{\xspace}{\ensuremath{\mathcal X}}
\newcommand{\zspace}{\ensuremath{\mathcal Z}}
\newcommand{\Thetanat}{\ensuremath{\Theta_{\text{nat}}}}
\newcommand{\Thetamean}{\ensuremath{\Theta_{\text{mean}}}}

\newcommand{\ceil}[1]{{\left\lceil#1\right\rceil}}
\newcommand{\floor}[1]{{\left\lfloor#1\right\rfloor}}

\newcommand{\di}{{\delta_i}}
\newcommand{\moment}[1]{\textnormal{m}_P^{(#1)}}

\newcommand{\integers}{{\mathbb{Z}}}
\newcommand{\reals}{{\mathbb{R}}}
\newcommand{\fractions}{{\mathbb{Q}}}
\newcommand{\nats}{{\mathbb{N}}}

\newcommand{\var}{{\textnormal{var}}}
\newcommand{\mhat}{{\hat\mu}}
\newcommand{\mihat}{{{\hat\mu}_i}}
\newcommand{\mnhat}{{{\hat\mu}_n}}
\newcommand{\mbar}{{\bar\mu}}
\newcommand{\mibar}{{{\bar\mu}_i}}
\newcommand{\mnbar}{{{\bar\mu}_n}}
\newcommand{\mstar}{{\mu^*}}
\newcommand{\mdagger}{{\mu^\dagger}}

\newcommand{\Uhat}{\hat{U}}

\newcommand{\mmstar}{{M_\mstar}}
\newcommand{\mmhat}{{M_\mhat}}
\newcommand{\mmihat}{{M_\mihat}}
\newcommand{\mmnhat}{{M_\mnhat}}
\newcommand{\mmu}{{M_\mu}}

\newcommand{\mitilde}{{\tilde\mu}_i}
\newcommand{\mntilde}{{\tilde\mu}_n}

\renewcommand{\lg}{\ln}

\newcommand\PP{P_{\textsc{\small p}}} \newcommand\PG{P_{\textsc{\small g}}}

\newcommand{\commentout}[1]{}

\newtheorem{definition}{Definition}
\newtheorem{example}{Example}
\newtheorem{theorem}{Theorem}
\newtheorem{corollary}{Corollary}
\newtheorem{lemma}[theorem]{Lemma}
\newtheorem{proposition}[theorem]{Proposition}
\newtheorem{conjecture}[theorem]{Conjecture}
\newtheorem{condition}{Condition}

\title{Prequential Plug-In Codes that Achieve Optimal Redundancy Rates even
  if the Model is Wrong}

\author{
  \authorblockN{Peter Gr\"unwald \texttt{pdg@cwi.nl}\qquad\qquad Wojciech
  Kot{\l}owski \texttt{kotlowsk@cwi.nl}} \authorblockA{%
    National Research Institute for Mathematics and Computer Science (CWI)\\
    P.O. Box 94079, 1090 GB Amsterdam, The Netherlands}
} \maketitle

\begin{abstract}
  We analyse the prequential plug-in codes relative to one-parameter
  exponential families $\model$. We show that if data are sampled
  i.i.d. from some distribution outside $\model$, then the redundancy of
  any plug-in prequential code grows at rate larger than
  ${1\over2} \ln n$ in the worst case. This means that
  plug-in codes, such as the Rissanen-Dawid ML code, may behave inferior
  to other important universal codes such as the 2-part MDL, Shtarkov
  and Bayes codes, for which the redundancy is always ${1\over2} \ln
  n + O(1)$. However, we also show that a slight modification of the ML
  plug-in code, ``almost'' in the model, does achieve the optimal
  redundancy even if the the true distribution is outside ${\cal M}$. 
\end{abstract}

\section{Introduction}
We resolve two open problems from \cite{Grunwald07} concerning
universal codes of the predictive plug-in type, also known as
``prequential'' codes. These codes were introduced independently by
Rissanen \cite{Rissanen84} in the context of MDL learning and by Dawid
\cite{Dawid84}, who proposed them as probability forecasting
strategies rather than directly as codes.  Roughly, the plug-in codes
relative to parametric model ${\cal M} = \{ M_{\theta} \mid \theta \in
\Theta \}$ work by sequentially coding each outcome $x_i$ based on an
an estimator $\bar{\theta}_{i-1} = \bar{\theta}(x^{i-1})$ for all
previous outcomes $x^{i-1} = x_1, \ldots, x_{i-1}$, leading to
codelength (log loss) $- \lg M_{\bar{\theta}_{i-1}}(x_i)$, where $M_{\theta}$
denotes the probability density or mass function indexed by $\theta$. If we take
$\bar{\theta}_i = \hat{\theta}_i$ equal to the ML (maximum likelihood)
estimator, we call the resulting code the ``ML plug-in code''.

There are many papers about the redundancy and/or expected regret for
the ML plug-in codes, for a large variety of models including
multivariate exponential families, ARMA processes, regression models
and so on. Examples are \cite{Rissanen1986c,Gerencser1987, LiY00}. In all these papers the ML plug-in
code is shown to achieve an asymptotic expected regret or redundancy
of ${k\over2}\lg n + O(1)$, where $k$ is the number of parameters of the
model and $n$ is the sample size. This matches the behaviour of the
Shtarkov, Bayesian and two-part universal codes
and is optimal in several ways, see \cite{BarronRY98}; since the ML
plug-in codes are often easier to calculate than any of these other
three codes, this appears to be a strong argument for using them in
practical data compression and MDL-style model selection. Yet, 
more recently \cite{DeRooijG05a,GrunwaldD05,DeRooijG06}, it was shown
that, at least for single-parameter exponential family models, when the data are generated i.i.d. $\sim P$, the
redundancy in fact grows as ${1\over2}\lg n\cdot{\var_P X\over\var_M
  X}$, where  $M$ is the distribution in $\model$ that is closest to
$P$ in Kullback-Leibler divergence, i.e. it minimizes $D(P\|M)$; a
related result for linear regression is in \cite{Wei1990}. In contrast
to the other cited works, \cite{DeRooijG05a,GrunwaldD05,DeRooijG06,Wei1990} do
not assume that $P \in \model$: the model may be {\em misspecified}.
Yet {\em if\/} $P\in\model$, then we have $M=P$ so that the redundancy
grows like it does in the other universal models. But when $M \neq P$,
the Shtarkov, Bayes and universal codes typically still achieve
asymptotic expected regret ${ 1 \over 2 } \lg n$, whereas the plug-in
codes behave differently. \cite{DeRooijG05a,DeRooijG06} show that this
leads to substantially inferior performance of the plug-in codes in
practical MDL model selection.
\subsection{The Two Open Problems/Conjectures}
In general, the estimator for ${\cal M}$ based on $x^{i-1}$ need not be an
element of the parametric model ${\cal M}$; for example, we may think
of the Bayesian predictive distribution as an estimator relative to
${\cal M}$, even though it is ``out-model'': rather than a single element of ${\cal
  M}$, it is a mixture of
distributions in ${\cal M}$, each weighted by their posterior density
(see Section~\ref{sec:squashed} for an example). We may thus re-interpret Bayesian universal codes as
prequential codes based on ``out-model'' estimators. From now on, we
reserve the term ``prequential plug-{\em in\/} code'', abbreviated to
just ``plug-in code'', for codes based on ``{\em in}-model''
estimators, i.e. estimators required to lie within ${\cal
  M}$. When we call a code just ``prequential'', it may be
sequentially constructed from either in-model or out-model
estimators. 
\cite{GrunwaldD05} established a nonstandard redundancy, different
from $(k/2) \lg n$, only for ML and closely related plug-in
codes. \cite[Open Problem Nr. 2]{Grunwald07} conjectured that a
similar result should hold for {\em all\/} plug-in codes, even if they
are based on in-model estimators very different from the ML estimator:
the conjecture was that {\em no\/}  plug-in code can achieve
guaranteed redundancy of $(k/2) \lg n$ if data are i.i.d.  $\sim P$
and $P \neq M$. Our first main result, Theorem~\ref{thm:no_plugin}
below, shows that, essentially, this conjecture is true for general
one-parameter exponential families $(k=1)$.  Specifically, the
redundancy can become much larger than $(1/2) \lg n$ if $P \not \in
{\cal M}$.

The second related conjecture \cite[Open Problem Nr. 3]{Grunwald07}
concerned the fact that for the normal location family with constant
variance $\sigma^2$, the Bayesian predictive distribution based on
data $x^{i-1}$ and a normal prior looks ``almost'' like an in-model
estimator for $x^{i-1}$, and hence the resulting code looks ``almost''
like a plug-in code: the Bayes predictive distribution is equal to the
normal distribution for $X_i$ with mean equal to the ML estimator
$\hat{\mu}(x^{i-1})$ but with a variance of order $\sigma^2 + O(1/n)$,
i.e. slightly larger than the variance $\sigma^2$ of
$P_{\hat{\mu}(x^{i-1})}$ (see Section~\ref{sec:squashed} for details).
Since the Bayesian predictive distribution does achieve the redundancy
$(1/2) \lg n$ even if $P \not \in {\cal M}$, this means that if ${\cal
  M}$ is the normal location family, then there does exist an
``almost'' in-model estimator (i.e. a slight modification of the ML
estimator) that does achieve $(1/2) \lg n$ even if $P \not \in {\cal
  M}$. Although this example does not extend straightforwardly to
other exponential families, \cite{Grunwald07} conjectured that there
should nevertheless be some general definition for ``almost'' in-model
estimators that achieve $(k/2) \lg n$ redundancy even if $P \not \in
{\cal M}$.  Here we show that this conjecture is true, at least if
$k=1$: we propose the {\em slightly squashed\/} ML estimator, a modification of
the ML estimator that puts it slightly outside model ${\cal M}$, and
in Theorem~\ref{thm:robustML} we show that this estimator achieves
$(1/2) \lg n$ redundancy even if $P \not \in {\cal M}$. This result is
important in practice since, in contrast to the
Bayesian predictive distribution, the slightly squashed ML estimator is in
general just as easy to compute as the ML estimator itself.

\section{Notation and Definitions}
Throughout this text we use nats rather than bits as units of
information.  
A sequence of outcomes $z_1,\ldots,
z_n$ is abbreviated to $z^n$. We write $E_{P}$ as a shorthand for $E_{Z \sim P}$,
the expectation of $Z$ under distribution $P$. When we consider a sequence of $n$
outcomes independently distributed $\sim P$, we use $E_{P}$ even as a shorthand
for the expectation of $(Z_1, \ldots, Z_n)$ under the $n$-fold product
distribution of $P$.  Finally, $P(Z)$ denotes the probability mass function of
$P$ in case $Z$ is discrete-valued, and it denotes the density of $P$, in case
$Z$ takes its value in a continuum. When we write `density function of $Z$',
then, if $Z$ is discrete-valued, this should be read as `probability mass
function of $Z$'. Note however that in our second main result, Theorem~\ref{thm:robustML} we do not
assume that the data-generating distribution $P$ admits a density.

Let $\zspace$ be a set of outcomes, taking values either in a finite or countable
set, or in a subset of $k$-dimensional Euclidean space for some $k \geq 1$. Let
$X: \zspace \rightarrow \reals$ be a random variable on $\zspace$, and let
$\xspace = \{ x \in \reals \; : \; \exists z \in \zspace: X(z) = x\}$ be the
range of $X$. Exponential family models are families of distributions on
$\zspace$ defined relative to a random variable $X$ (called `sufficient
statistic') as defined above, and a function $h: \zspace \rightarrow [0,\infty)$.
Let $\partf(\eta) :=\int_{z\in\zspace}e^{-\eta X(z)}h(z)dz$ (the integral to be
replaced by a sum for countable $\zspace$), and  $\Thetanat
:=\{\eta\in\reals:\partf(\eta)<\infty\}$.

\begin{definition}[Exponential family]\label{def:expfam}
  The \emph{single parameter exponential family} {\rm \cite{BarndorffNielsen78}}
  with {\em sufficient statistic\/} $X$ and {\em carrier\/} $h$ is the
  family of distributions with densities $M_\eta(z) :={1\over
    \partf(\eta)}e^{-\eta X(z)}h(z)$, where $\eta\in \Thetanat$.
  $\Thetanat$ is called the \emph{natural parameter space}. The
  family is called \emph{regular} if $\Thetanat$ is an open interval
  of $\reals$.
\end{definition}

In the remainder of this text we only consider single parameter, regular
exponential families, but this qualification will
henceforth be omitted. Examples include the Poisson, geometric and multinomial
families, and the model of all Gaussian  distributions with a fixed variance or 
mean.

The statistic $X(z)$ is sufficient for $\eta$ \cite{BarndorffNielsen78}. This suggests
reparameterizing the distribution by the expected value of $X$, which is called
the \emph{mean value parameterization}. The function $\mu(\eta)=E_{M_\eta}[X]$
maps parameters in the natural parameterization to the mean value
parameterization. It is a diffeomorphism (it is one-to-one, onto, infinitely
often differentiable and has an infinitely often differentiable inverse)
\cite{BarndorffNielsen78}. Therefore the mean value parameter space $\Thetamean$ is also an
open interval of $\reals$. We write $\model=\{\mmu\mid\mu\in\Thetamean\}$ where
$\mmu$ is the distribution with mean value parameter $\mu$. 
\commentout{We note that for some
models (such as Bernoulli and Poisson), the parameter space is usually given in
terms of the a non-open set of mean-values (e.g., $[0,1]$ in the Bernoulli case).
To make the model a regular exponential family, we have to restrict the set of
parameters to its own interior. Henceforth, whenever we refer to a standard
statistical model such as Bernoulli or Poisson, we assume that the parameter set
has been restricted in this sense.
}

We are now ready to define the plug-in universal model. This is a
distribution on infinite sequences $z_1, z_2, \ldots \in
\zspace^{\infty}$, recursively defined in terms of the distributions
of $Z_{n+1}$ conditioned on $Z^n = z^n$, for all $n = 1,2 , \ldots$.
In the definition, we use the notation $x_i := X(z_i)$. Note that we
use the term ``model'' both for a single distribution (``plug-in
universal model'', a common phrase in information theory) and for a
family of distributions (``statistical model'', a common phrase in
statistics).

\begin{definition}[Plug-in universal model]
\label{def:preq}
Let $\model=\{\mmu\mid\mu\in \Thetamean \}$ be an exponential family with mean
value parameter domain $\Thetamean$.  Given $\model$, constant $\mbar_0 \in
\Thetamean$ and a sequence of functions $\mbar(z^1),\mbar(z^2),\ldots$, such that
$\mbar(z^n) =: \mnbar  \in \Thetamean$, we define the
\emph{plug-in universal model} (or \emph{plug-in model} for short) $U$ by setting, for all $n$, all $z^{n+1}
\in \zspace^{n+1}$: $$U(z_{n+1} \mid z^n) = M_{\mnbar}(z_{n+1}),$$ where
$U(z_{n+1} \mid z^n)$ is the density/mass function of $z_{n+1}$ conditional on
$Z^n = z^n$. \end{definition}\smallskip

We usually refer to  plug-in universal model in terms of the 
codelength function of the corresponding plug-in universal code: 
\begin{equation}
\label{eq:codelength}
L_U(z^n) = \sum_{i=0}^{n-1}  L_U(z_{i+1} \mid z_i) =
\sum_{i=0}^{n-1} - \ln M_{\mibar}(z_{i+1}).
\end{equation}

The most important plug-in model is the ML (\emph{maximum
likelihood}) plug-in model, defined as follows:

\begin{definition}[ML plug-in model]
\label{def:preqb} Given $\model$ and
constants $x_0\in\Thetamean$ and $n_0>0$, we define the
\emph{ML plug-in model} $\Uhat$ by setting, for all $n$, all $z^{n+1}
\in \zspace^{n+1}$:
$$\Uhat(z_{n+1} \mid z^n) = M_{\hat{\mu}(z^n)}(z_{n+1}),$$
where
\begin{equation}
\label{eq:ML_estimator}
\hat{\mu}(z^n) = \mnhat :=\frac{x_0 \cdot n_0+\sum_{i=1}^n
x_i}{n+n_0}. 
\end{equation}
\end{definition}\smallskip

To understand this definition, note that for exponential families, for any
sequence of data, the ordinary maximum likelihood parameter is given by the
average $n^{-1} \sum x_i$ of the observed values of $X$ \cite{BarndorffNielsen78}.  Here we
define our plug-in model in terms of a slightly modified maximum likelihood
estimator that introduces a `fake initial outcome' $x_0$ with multiplicity $n_0$
in order to avoid infinite code lengths for the first few outcomes (a well-known
problem sometimes called the ``inherent singularity'' of predictive coding
\cite{BarronRY98,Grunwald07}) and to ensure that the plug-in ML code
of the first outcome is well-defined. In practice we can take $n_0 = 1$ but our
result holds for any $n_0 > 0$.

\begin{definition}[Relative redundancy]
 \label{def:red}
Following \cite{TakeuchiB98b,DeRooijG05a}, we define {\em relative redundancy\/} with respect to $P$ of a code $U$ that is
universal on a model $\model$, as:
\begin{equation}\label{eq:red}
\redundancy_{U}(n) := E_P[L_U(Z^n)] - \inf_{\mu \in \Thetamean}E_P [-
  \ln \mmu(Z^n)],
\end{equation}
where $L_U$ is the length function of $U$.
\end{definition}

We use the term \emph{relative redundancy} rather than just \emph{redundancy} to
emphasize that it measures redundancy relative to the element of the model that
minimizes the codelength rather than to $P$, which is not necessarily an element
of the model. From now on, we only consider $P$ under which the data
are i.i.d.  Under this condition, let $\mmstar$ be the element of
$\model$ that minimizes KL divergence to $P$: $$\mstar := \arg \min_{{\mu} \in
\Thetamean} D(P \| M_{\mu}) = \arg \min_{{\mu} \in \Thetamean} E_{P} [ - \ln
\mmu(Z)],$$
where the equality follows from the definition of KL divergence. If
$\mmstar$ exists, it is unique, and if $E_P[X] \in \Thetamean$, then
$\mu^* = E_P[X]$ \cite[Ch. 17]{Grunwald07}, and the relative redundancy satisfies
\begin{equation}
\label{eq:redb}
\redundancy_{U}(n) = E_{P}[L_U(Z^n)] - 
E_{P}[- \ln \mmstar(Z^n)]. 
\end{equation}
\section{First Result: Redundancy of Plug-In Codes}\label{sec:mainresult}
The three major types of universal codes, Bayes, NML and 2-part, achieve relative
redundancies that are (in an appropriate sense) close to optimal. Specifically,
under the conditions on $\model$ described above, and if data are
i.i.d. $\sim P$, then, under some mild conditions
on $P$, these universal codes satisfy:
\begin{equation}
\label{eq:bic}
\redundancy_U(n) = \frac{1}{2} \ln n + O(1),
\end{equation}
(where the $O(1)$ may depend on $\mu$ and the universal code
used), whenever $P \in \model$ or $P \not \in \model$.
(\ref{eq:bic}) is the famous `$k$ over $2$ log $n$ formula' ($k=1$ in our case),
refinements of which lie at the basis of practical
approximations to MDL learning \cite{Grunwald07}. 

While it is known that for $P \in \model$, the fourth major type of universal
code, the ML plug-in code, satisfies (\ref{eq:bic}) as well, it was
shown by 
\cite{DeRooijG05a,GrunwaldD05} that when $P$ is not in the model, the ML
plug-in code
may behave suboptimally. Specifically, its relative redundancy satisfies:
\begin{equation}
\label{eq:redundancyML}
\redundancy_{\Uhat}(n) = \frac{1}{2} \frac{\var_P X}{\var_{M_{\mu^*}} X} \ln n +
O(1),
\end{equation}
and can be significantly larger than (\ref{eq:bic}), when the variance of $P$ is large.

In this paper, we show that not only the ML plug-in code, but
\emph{every} plug-in code may behave suboptimally, when $P \notin
\model$. In other words, modifying the ML estimator $\mnhat$ or
introducing any other sequence of estimators $\mnbar$, and
constructing the plug-in code based on that sequence will not help to
satisfy (\ref{eq:bic}). Thus the optimal redundancy can only be
achieved by codes outside $\model$, unless $\model$ is the
Bernoulli family (since we assume the data are i.i.d., in the
Bernoulli case we must have that $P \in {\cal M}$; but the Bernoulli
case is the {\em only\/} case in which we must have $P \in {\cal M}$).

Our main result, Theorem \ref{thm:no_plugin}, concerns the case in
which $P$ is itself a member of some exponential family $\gen$, but
$\gen$ is in general different than $\model$. Then, the suboptimal
behavior of plug-in codes follows immediately as  Corollary
\ref{coll:no_plugin}, stated further below.

\begin{theorem}
\label{thm:no_plugin}
Let $\model = \{\mmu\mid\mu\in\Thetamean\}$ and $\gen = \{P_\mu \mid \mu \in \Thetamean\}$ be single parameter exponential families with the same sufficient
statistic $X$ and mean-value parameter space $\Thetamean$. Let $U$ denote
any plug-in model with respect to $\model$ based on the sequence
of estimators $\mbar_0,\mbar_1,\mbar_2,\ldots$.
Then, for Lebesgue almost all
$\mstar \in \Thetamean$ (i.e. all apart from a Lebesgue measure zero set), for $X, X_1, X_2, \ldots$
 i.i.d. $\sim
P_\mstar \in \gen$:
$$
\liminf_{n \to \infty} \frac{\mathcal{R}_U(n)}{\frac{1}{2} \ln n} \geq \frac{\var_{P_\mstar} X}{\var_{M_\mstar} X}.
$$
\end{theorem}

\begin{proof}{\em (rough sketch; a detailed proof is in the Appendix)}
The proof is based on a theorem stated by Rissanen \cite{Rissanen86a} (see also
\cite{Grunwald07}, Theorem 14.2), a special case of which says the following.
Let $\Theta_0 \subset \Thetamean$ be a closed, non-degenerate interval, $\gen$ be defined as above, $P_{\mu}^{(n)}$ be a joint
distribution of $n$ outcomes generated i.i.d. from $P_\mu$, $Q$ be an arbitrary
probabilistic source, i.e. a distribution on infinite sequences $z_1, z_2, \ldots
\in \zspace^{\infty}$, and let $Q^{(n)}$ be its restriction to the first $n$
outcomes. Define:
$g_n(\mstar) = \frac{D(P^{(n)}_\mstar \|
Q^{(n)})}{\frac{1}{2} \ln n}.$
Then for Lebesgue almost all  $\mstar \in \Theta_0$, $\liminf_{n \to \infty}
g_n(\mstar) \geq 1$. 

We apply Rissanen's theorem by constructing a source $Q$, specifying the conditional probabilities
$Q(z_{n+1}|z^{n}) := P_\mnbar$, for every $n \geq 1$. We now have:
\begin{eqnarray}
\label{eq:KL_deriv}
D(P_\mstar^{(n)} \| Q^{(n)})\!\!\!&=&\!\!\! \sum_{i=0}^{n-1} E_{P_\mstar}
\left[ \ln P_{\mstar}(Z_{i+1}) - \ln Q(Z_{i+1}|Z^{i}) \right] \nonumber \\
\!\!\!&=&\!\!\! \sum_{i=1}^{n-1} E_{P_\mstar} \left[ D(P_{\mstar} \|
P_{\mibar}) \right].
\end{eqnarray}
To see how (\ref{eq:KL_deriv}) is related to our case, let us first rewrite the redundancy in a more convenient form: 
\begin{equation}
\label{eq:redundancy_convienient}
\redundancy_U(n)=\sum_{i=0}^{n-1}\,E_{P_\mstar}\left[D(\mmstar\parallel
M_\mibar)\right].
\end{equation}
The derivation of (\ref{eq:redundancy_convienient}) make use of a standard result in the theory of exponential families and can be found e.g. in \cite{Grunwald07}.

Comparing (\ref{eq:KL_deriv}) and (\ref{eq:redundancy_convienient}),
we see that although in both expressions, the expectation is taken
with respect to $P_\mstar$, (\ref{eq:KL_deriv}) is a statement about
KL divergence between the members of $\gen$, while
(\ref{eq:redundancy_convienient}) speaks about the members of
$\model$. The trick, which allows us to relate both expressions, is to
examine their second-order behavior. By expanding $D(P_{\mstar} \|
P_{\mibar})$ into a Taylor series around $\mstar$, we get: $$
D(P_{\mstar} \| P_{\mibar}) \simeq 0 + \expdif{1}(\mstar) (\mibar -
\mstar) + \frac{1}{2} \expdif{2}(\mstar) (\mibar - \mstar)^2,
$$
where we abbreviated $\expdif{k}(\mu) = \frac{d^k}{d\mu^k} D(P_\mstar\|P_\mu)$. The term $\expdif{1}(\mstar)$ is
zero, since $D(\mstar \| \mu)$ as a function of $\mu$ has its
minimum at $\mu=\mstar$ \cite{BarndorffNielsen78}. As is well-known \cite{BarndorffNielsen78}, for
exponential families the term $\expdif{2}(\mu)$ coincides precisely with the
Fisher information $I_{\gen}(\mu)$ evaluated at $\mu$. Another standard result
\cite{BarndorffNielsen78} for the mean-value parameterization says
that for all $\mu$, $I_{\gen}(\mu) = \frac{1}{\var_{P_\mu}
  X}$. Therefore, we get $D(P_{\mstar} \| P_{\mibar}) \simeq
\frac{1}{2} \frac{(\mibar - \mstar)^2}{\var_{P_\mstar} X}$, and
similarly, $D(M_{\mstar} \| M_{\mibar}) \simeq \frac{1}{2}
\frac{(\mibar - \mstar)^2}{\var_{M_\mstar} X}$, so that 
$
D(M_{\mstar} \| M_{\mibar}) \simeq D(P_{\mstar} \| P_{\mibar}) \frac{\var_{P_\mstar} X}{\var_{M_\mstar} X}
$, 
and using (\ref{eq:KL_deriv}) and (\ref{eq:redundancy_convienient}):
$$
\begin{array}{c}
\redundancy_U(n) \simeq D(P_\mstar^{(n)} \| Q^{(n)}) \frac{\var_{P_\mstar} X}{\var_{M_\mstar} X}.
\end{array}
$$
The last step of the proof is to use Rissanen's theorem and conclude
that
$\liminf_{n \to \infty} \frac{\redundancy_U(n)}{\frac{1}{2}  \ln n}$
is equal to
$$
\begin{array}{c}
\liminf_{n \to \infty} \frac{D(P^{(n)}_\mstar \|
Q^{(n)})}{\frac{1}{2} \ln n} \frac{\var_{P_\mstar} X}{\var_{M_\mstar}
X} 
\geq \frac{\var_{P_\mstar} X}{\var_{M_\mstar} X},
\end{array}
$$
for Lebesgue almost all  $\mstar \in \Theta_0$, and thus for Lebesgue almost all  $\mstar \in \Thetamean$.
\end{proof}
We now use Theorem~\ref{thm:no_plugin}
to show that the redundancy of plug-in codes is suboptimal for all
exponential families which satisfy the following very weak condition: 
\begin{condition}\label{cnd:dispersion_model}
Let $\model = \{\mmu\mid\mu\in\Thetamean\}$ be a single parameter exponential
family with sufficient statistic $X$ and mean-value parameter space $\Thetamean$.
We require that there exists another single-parameter exponential family $\gen = \{P_\mu \mid \mu \in \Thetamean\}$ with the same
mean-value parameter space as $\model$, but with strictly larger variance 
than $\model$ for every $\mu \in \Thetamean$.
\end{condition}

The Condition \ref{cnd:dispersion_model} is widely satisfied among known
exponential families. When $\xspace = [a,b]$, we define
$P_\mu$ to be a ``scaled'' Bernoulli model, by putting all probability mass on
$\{a,b\}$ in such a way that $E_{P_\mu} =\mu$. It is easy to show, that
such distribution has the highest variance among all distributions defined on $[a,b]$ with a given
mean value $\mu$; therefore $\var_{P_\mu}X > \var_{M_\mu}X$, unless $\model$ is
a ``scaled'' Bernoulli itself. When $\xspace = \reals$, $\gen$ can be chosen
 to be a normal family with fixed, sufficiently large variance
$\sigma^2$. For $X=[0,\infty)$, $\gen$ can be taken to be a gamma family with sufficiently
large scale parameter. When $\xspace = \{0,1,2,\ldots\}$, $\gen$ can
be taken to be 
negative binomial (with expected ``number of successes'' sufficiently small). 

Thus, we see that for all commonly
used exponential families, except for Bernoulli, Condition \ref{cnd:dispersion_model} holds. On the
other hand if $\model$ is Bernoulli, Corollary \ref{coll:no_plugin}
is no longer relevant anyway, since then $P$ must lie in ${\cal M}$.
\begin{corollary}
\label{coll:no_plugin}
Let $\model = \{\mmu\mid\mu\in\Thetamean\}$ a single parameter exponential
family with sufficient statistic $X$ and mean-value parameter space
$\Thetamean$, satisfying Condition \ref{cnd:dispersion_model}. Let $U$ denote
any plug-in model with respect to $\model$ based on any sequence
of estimators $\mbar_1,\mbar_2,\ldots$.
Then, there exists a family of distributions $\gen = \{P_\mu \mid \mu \in
\Thetamean\}$, such that for Lebesgue almost all
$\mstar \in \Thetamean$, for $X, X_1, X_2, \ldots$  i.i.d. $\sim
P_\mstar$:
$$
\liminf_{n \to \infty} \frac{\mathcal{R}_U(n)}{\ln n} 
\geq  \frac{1}{2} \frac{\var_{P_\mstar} X}{\var_{M_\mstar} X}
> \frac{1}{2},
$$
so that the set of $\mstar$ for which $U$ achieves the regret $\frac{1}{2} \ln n
+ O(1)$ is a set of Lebesgue measure zero.
\end{corollary}
\begin{proof}
Immediate from Theorem \ref{thm:no_plugin}.
\end{proof}

\section{Second  Result: Optimality of  Squashed ML}
\label{sec:squashed}
We showed that every plug-in code, including the ML plug-in code,
behaves suboptimally for 1-parameter families $\model$ unless $\model$
is Bernoulli. This fact does not, however, exclude the possibility
that a small modification of the ML plug-in code, which puts the
predictions slightly outside $\model$, will lead to the optimal redundancy
(\ref{eq:bic}).  An argument supporting this claim comes from
considering the Bayesian predictive distribution when $\model$ is the normal family with
fixed variance $\sigma^2$. In this case, the Bayesian code based on prior
$\mathcal{N}(\mu_0,\tau^2_0)$ has a simple form \cite{Grunwald07}:
$$U_{\text{Bayes}}(z_{n+1} \mid z^n) = f_{\mu_n,\tau_n^2+\sigma^2}(z_{n+1}),$$
where $f_{\mu,\sigma^2}$ is the density of normal distribution
$\mathcal{N}(\mu,\sigma^2)$,
$$ 
\begin{array}{c}\mu_n = \frac{(\sum_{i=1}^n x_i) + \frac{\sigma^2}{\tau_0^2}
\mu_0}{n+\frac{\sigma^2}{\tau_0^2}}, 
\quad \text{and} \quad 
\tau_n^2 = \frac{\sigma^2}{n + \frac{\sigma^2}{\tau_0}}.\end{array} $$
Thus, the Bayesian predictive distribution is itself a Gaussian with
mean equal to the modified maximum likelihood estimator (with 
$n_0 = \sigma^2/\tau_0^2$), albeit
with a slightly larger variance $\sigma^2 + O(1/n)$. This shows that
for the normal family with fixed variance, there exists an ``almost'' in-model code, which 
satisfies (\ref{eq:bic}). This led \cite{Grunwald07} to conjecture
that something similar holds for general exponential families. Here we
show that this is indeed the case: we propose a simple modification of the ML plug-in universal
model, obtained by predicting $z_{n+1}$ using a slightly ``squashed'' version
$M'_{\hat{\mu}_n}$ of
the ML estimator $M_{\hat{\mu}_n}$, defined as:
$$
M'_{\mnhat}(z_{n+1})
: = M_{\mnhat}(z_{n+1}) \frac{1 +
\frac{1}{2n} I_\model(\mnhat)(x_{n+1} - \mnhat)^2}{1 + \frac{1}{2n}},
$$
where
$\mnhat$ is defined as in (\ref{eq:ML_estimator}) and $I_\model(\mu)$ is the
Fisher information for model $\model$.
Note that $M'_{\mnhat}(z_{n+1})(\cdot)$ represents a
valid probability density: it is non-negative due to
$I_\model(\mnhat) > 0$ (property of exponential families), and it is
properly normalized:
$$
\begin{array}{cc}
&\int_\xspace M'_{\mnhat}(z_{n+1})(z) dz =  
\left(1 + \frac{1}{2n} \right)^{-1} \bigg( \int_\xspace M_{\mnhat}(z) dz \\
&+ 
\frac{1}{2n}I_\model(\mnhat) \int_\xspace (X(z) - \mnhat)^2 M_{\mnhat}(z) dz
 \bigg) = 1,
\end{array}
$$
where the final equality follows because for exponential families,
$I_{\cal M}(\mu) = (\var_{M_{\mu}}X)^{-1}$. 
While $M' \not \in {\cal M}$, we have $D(M'_{\hat{\mu}_n} \|
M_{\hat{\mu}_n}) = O(1/n)$, i.e. $M'$ is ``almost'' in-model
estimator.
\begin{definition}[Squashed ML prequential model]
\label{def:preqc} Given $\model$,
constants $x_0\in\Thetamean$ and $n_0>0$, we define the
\emph{slightly squashed ML prequential model} $U$ by setting, for all $n$, all
$z^{n+1} \in \zspace^{n+1}$:
$$U(z_{n+1} \mid z^n) = M'_{\mnhat}(z_{n+1}),$$ 
\end{definition}\smallskip
where $M'$ is the slightly squashed ML estimator as above. 
The codelengths of the corresponding slightly squashed ML prequential code are
not harder to calculate than those of the ordinary
ML plug-in model and in some cases they are  easier to calculate than the
lengths of the Bayesian universal code. On the other hand, we show
below that the slightly squashed ML code always achieves the optimal
redundancy, satisfying (\ref{eq:bic}). 

\begin{theorem}
\label{thm:robustML}
  Let $X, X_1, X_2, \ldots$ be i.i.d.$\sim P$, with $E_P[X] = \mstar$. 
Let  $\model$ be a single parameter exponential family with sufficient
statistic $X$ and $\mstar$ an element of the mean value parameter space. Let
$U$ denote the slightly squashed ML model with respect to $\model$.  If
$\model$ and $P$ satisfy Condition \ref{condition} below,
then:
\begin{equation}
\label{eq:main}
\redundancy_U(n)= \frac{1}{2} \ln n + O(1).
\end{equation}
\end{theorem}

\begin{condition}\label{condition}
We require that the following holds both for $T:= X$ and $T:=- X$:
\begin{itemize}
\item If $T$ is unbounded from above then there is a
  $k\in\{4,6,\ldots\}$ such that the first $k$ moments of $T$ exist under $P$,
  that ${d^2\over d\mu^2}I_\model(\mu)=O\left(\mu^{k-4}\right)$, ${d^4 \over
  d\mu^4} D(M_\mstar \| \mmu) = O(\mu^{k-6})$ and that either $I_\model(\mu)$ is
  constant or $I_\model(\mu)=O\left(\mu^{k/2-3}\right)$.
\item If $T$ is bounded from above by a constant $g$ then ${d^2\over
d\mu^2}I_\model(\mu)$, ${d^4 \over
  d\mu^4} D(M_\mstar \| \mmu)$, and $I_\model(\mu)$ are polynomial in
  $1/(g-\mu)$.
\end{itemize}
\end{condition}
The usefulness of Theorem \ref{thm:robustML} depends on the validity of
Condition \ref{condition} among commonly used exponential families. As
can be seen from Figure
\ref{fig:condition}, for some standard exponential families, 
our condition applies whenever the fourth moment of $P$ exists.
\begin{proof}{\em (of Theorem~\ref{thm:robustML}; rough sketch --- a
    detailed proof is in the Appendix)}
We express the relative redundancy of the slightly squashed
ML plug-in code $U$ by the sum of the relative redundancy of the ordinary
ML plug-in code $\Uhat$ and the difference in expected codelengths between
$U$ and $\Uhat$:
$$
\begin{array}{l}
\redundancy_{U}(n) = E_{P}[L_U(Z^n)] - 
E_{P}[- \ln \mmstar(Z^n)] =\\
E_{P}[L_U(Z^n) - L_{\Uhat}(Z^n)] + \redundancy_{\Uhat}(n)
= \\  E_{P}[L_U(Z^n) - L_{\Uhat}(Z^n)] + \frac{1}{2}
\frac{\var_P X}{\var_{M_{\mu^*}} X} \ln n + O(1),
\end{array}
\vspace*{-3pt}
$$
where the last equality follows from (\ref{eq:redundancyML}).
We have:
$$
\begin{array}{l}
L_U(Z^n) - L_{\Uhat}(Z^n) = \\
\sum_{i=0}^{n-1} \left(- \ln U(Z_{i+1} \mid Z_i) +
\ln {\Uhat}(Z_{i+1} \mid Z_i) \right)  =  \\
\sum_{i=0}^{n-1} \left( \ln\left(1 + \frac{1}{2i}\right) -
\ln\left(1 + \frac{1}{2i} I_\model(\mihat)(X_{i+1} -
\mihat)^2\right)\right) .
\end{array}
$$
Since $\ln\left(1+\frac{1}{2i}\right) = \frac{1}{2i} + O(i^{-2})$, we get
$\sum_{i=0}^{n-1} \ln\left(1 + \frac{1}{2i}\right) = \frac{1}{2} \ln n + O(1)
$. Denoting $V_i =\frac{1}{2i}
I_\model(\mihat)(X_{i+1} - \mihat)^2$, we also get $\ln(1+V_i) = V_i +
O(i^{-2})$. Next, we consider $E_P[V_i]$:
$$
\begin{array}{l}
E_P[V_i] = \frac{1}{2i} E_P \left[ 
I_\model(\mihat)(X_{i+1} - \mstar + \mstar - \mihat)^2 \right] = \\
\frac{1}{2i} E_P \left[I_\model(\mihat)\left( \var_{P_\mstar}X + (\mstar - \mihat)^2
 \right) \right] 
= \\ \frac{1}{2i} \left(\var_{P_\mstar}X E_P \left[I_\model(\mihat)\right] + E_P
\left[I_\model(\mihat) (\mstar - \mihat)^2\right] \right).
\end{array}
$$
The second term $E_P\left[I_\model(\mihat) (\mstar - \mihat)^2\right]$ is
$O(i^{-1})$ as $E_P[(\mstar - \mihat)^2] = O(i^{-1})$ and $E[I_\model(\mihat)]
= I_\model(\mstar) + O(i^{-1})$ (follows from expanding $I_\model(\mihat)$ up to
the first order around $\mstar$). Similarly, the first term is
$(\var_{P_\mstar}X) I_\model(\mstar) + O(i^{-1})$. Thus, using
$I_\model(\mstar) = \frac{1}{\var_{M_\mstar}X}$, we finally get: $$ E_P[-\ln(1+V_i)] = -E_P[V_i] + O(i^{-2}) = \frac{1}{2i} \frac{\var_{P_\mstar}X}{\var_{M_\mstar}X} + O(i^{-2}). $$ Taking all together, we see that the terms $\frac{\var_{P_\mstar}X}{\var_{M_\mstar}X}$ cancel and we finally get
$
R_U(n) = \frac{1}{2} \ln n + O(1).
$
Condition \ref{condition} is necessary to ensure that all Taylor
expansions above hold.
\end{proof}

\begin{figure}
\caption{Fisher information, its second derivative and a fourth derivative of
the divergence for a number of exponential families. 
For the normal
distribution with fixed mean
we use mean 0 and the density of the squared outcomes is given as a function of the variance. 
}
\begin{center}
\label{fig:condition}
\begin{footnotesize}
\begin{tabular}{l@{}c@{}c@{}c}
\hline
Distribution & $I(\mu)$ & $\frac{d^2}{d\mu^2}I(\mu)$ & ${d^4 \over
  d\mu^4} D(M_\mstar \| \mmu)$ \\
\hline \\[-3mm] 
Bernoulli & $\frac{1}{\mu(1-\mu)}$ & $\frac{2}{\mu^3} + \frac{2}{(1-\mu)^3}$ &
$\frac{6\mstar}{\mu^4} + \frac{6(1-\mstar)}{(1-\mu)^4}$\\ 
Poisson & $\frac{1}{\mu}$ & $\frac{2}{\mu^3}$ & $\frac{6\mstar}{\mu^4}$\\
Geometric & $\frac{1}{\mu(\mu-1)}$ & $- \frac{2}{\mu^3} + \frac{2}{(1-\mu)^3}$ & 
$\frac{6\mstar}{\mu^4} - \frac{6(\mstar+1)}{(\mu+1)^4}$\\
Gamma (fixed $k$) & $\frac{k}{\mu^2}$ & $\frac{6k}{\mu^4}$ &
 $-\frac{6k}{\mu^4} + \frac{24k\mstar}{\mu^5}$\\ 
 Normal (fixed mean) & $\frac{1}{2\mu^2}$ & $\frac{3}{\mu^4}$ &
 $-\frac{3}{\mu^4} + \frac{12\mstar}{\mu^5}$ \\
 Normal (fixed variance) &  $\sigma^2$ & $0$ & $0$\\
\hline
\end{tabular}
\vspace*{-20pt}
\end{footnotesize}
\end{center}
\end{figure}
\section{Future Work}
In future work, we hope to extend our results concerning the slightly squashed
ML estimator to the multi-parameter case and establish almost-sure variation
of Theorem ~\ref{thm:robustML}. We also plan to analyze the estimator 
in the individual sequence framework, along the lines of
\cite{CesaBianchiLugosi06,Raginsky09}.
\bibliographystyle{IEEEtran}

\appendix
\section*{Proof of Theorem~\ref{thm:no_plugin}}
Before we show the main result, we need to prove the following lemmas.

\begin{lemma}
\label{lem:Rissanen_based}
Let $\model = \{\mmu\mid\mu\in\Thetamean\}$ and $\gen = \{P_\mu \mid \mu \in \Thetamean\}$ be
single parameter exponential families with the same sufficient
statistic $X$ and mean-value parameter space $\Thetamean$. Let
$\Theta_0 \subset \Thetamean$ be any non-degenerate closed interval.
Let $X, X_1, X_2, \ldots$ be i.i.d. $\sim P_\mstar$ for some $\mstar \in
\Theta_0$. Let $\mbar_0, \mbar_1,\mbar_2,\ldots$ be a  sequence of estimators, such that $\mibar = \mibar(z^i)$ and $\mibar \in \Theta_0$ for all $i
\geq 1$.  Then, for Lebesgue almost all $\mstar \in \Theta_0$:
$$
\liminf_{n \to \infty} \frac{\sum_{i=0}^{n-1} E_{P_\mstar}[(\mibar - \mstar)^2]
}{\ln n} \geq \underline{V}_{\gen} 
$$  
where $\underline{V}_{\gen} := \inf_{\mu \in \Theta_0}
\var_{P_\mu}X$.
\end{lemma}
\begin{proof}
The proof is based on a theorem stated by Rissanen \cite{Rissanen86a} (see also
\cite{Grunwald07}, Theorem 14.2), a special case of which says the following.

Let $\gen$ and $\Theta_0$ be defined as above, $P_{\mu}^{(n)}$ be a joint
distribution of $n$ outcomes generated i.i.d. from $P_\mu$, $Q$ be an arbitrary
probabilistic source, i.e. a distribution on infinite sequences $z_1, z_2, \ldots
\in \zspace^{\infty}$, and let $Q^{(n)}$ be its restriction to the first $n$
outcomes (marginalized over $z_{n+1}, z_{n+2},\ldots$). Define:
\begin{equation}
\label{eq:Rissanen_thm}
g_n(\mstar) = \inf_{n' \geq n} \left\{ \frac{D(P^{(n')}_\mstar \|
Q^{(n')})}{\frac{1}{2} \ln n'} \right\}.
\end{equation}
Then for Lebesgue almost all  $\mstar \in \Theta_0$, $\lim_{n \to \infty}
g_n(\mstar) \geq 1$.

We construct the source $Q$ by specifying the conditional probabilities:
$$
Q(z_{n+1}|z^{n}) := P_\mnbar,
$$
for every $n \geq 1$. This definition is valid, because $\mnbar$ depends only
on $z^{n}$. Now, we have:
\begin{eqnarray*}
&&D(P_\mstar^{(n)} \| Q^{(n)}) 
= E_{Z^{n} \sim P^{(n)}_{\mstar}} [\ln P_{\mstar}(Z^{n}) - \ln Q(Z^{n})] \\ 
&=& \sum_{i=0}^{n-1} E_{Z^{i} \sim P_\mstar^{(i)}} \left[ \ln
P_{\mstar}(Z_{i+1}) - \ln Q(Z_{i+1}|Z^{i}) \right] \\ 
&=& \sum_{i=1}^{n-1} E_{Z^{i} \sim P_\mstar^{(i)}} \left[ D(P_{\mstar} \|
P_{\mibar}) \right]. \\
\end{eqnarray*}
Expanding $D(P_{\mstar} \| P_{\mibar})$ into a Taylor series around $\mstar$ yeilds:
$$
D(P_{\mstar} \| P_{\mibar}) = 0 + \expdif{1}(\mstar) (\mibar - \mstar) +
\frac{1}{2} \expdif{2}(\mu) (\mibar - \mstar)^2,
$$
for some $\mu$ between $\mibar$ and $\mstar$, where we abbreviated
$\expdif{k}(\mu) = \frac{d^k}{d\mu^k} D(P_\mstar\|P_\mu)$. The term
$\expdif{1}(\mstar)$ is zero, since $D(\mstar \| \mu)$ as a function of $\mu$ has its
minimum at $\mu=\mstar$ \cite{BarndorffNielsen78}. As is well-known \cite{BarndorffNielsen78}, for
exponential families the term $\expdif{2}(\mu)$ coincides precisely with the
Fisher information $I_{\gen}(\mu)$ evaluated at $\mu$. Another standard result
\cite{BarndorffNielsen78} for the mean-value parameterization says that for all $\mu$,
\begin{equation}\label{eq:efficient}
I_{\gen}(\mu) = \frac{1}{\var_{P_\mu} X}.
\end{equation}
Therefore (using shorter notation $E_{P_\mstar}$ for $E_{Z^{i} \sim
P_\mstar^{(i)}}$):
 %
\begin{equation}
\begin{array}{ll}
&D(P_\mstar^{(n)} \| Q^{(n)}) 
= \frac{1}{2} \sum_{i=0}^{n-1} E_{P_\mstar} \left[
\frac{(\mibar - \mstar)^2}{\var_{P_\mu} X}   \right] \\
&\leq  \frac{1}{2} \frac{1}{\underline{V}_{\gen}} \sum_{i=0}^{n-1}
E_{P_\mstar} \left[ (\mibar - \mstar)^2 \right]. 
\label{eq:squared_approx}
\end{array}
\end{equation}
%
Note, that $\underline{V}_{\gen} > 0$ is an infimum of a
continuous and positive function on a compact set. From
(\ref{eq:Rissanen_thm}) and (\ref{eq:squared_approx}) we have:
$$
\inf_{n' \geq n} \left\{ \frac{\frac{1}{2} \sum_{i=0}^{n'-1} E_{P_{\mstar}}
\left[ (\mibar - \mstar)^2  \right]}{\frac{1}{2} \ln n'} \right\} \geq
 g_n(\mstar) \underline{V}_{\gen},
$$
and thus  Rissanen's theorem proves the lemma.
\end{proof}

\begin{lemma}
\label{lemm:variance_ratio}
Let $\model,\gen,\Theta_0,X,X_1,X_2,\ldots$, be defined as in Lemma
\ref{lem:Rissanen_based}. Let $U$ denote any plug-in model with
respect to $\model$ based on a sequence of estimators
$\mbar_1,\mbar_2,\ldots$ (notice that now we do not restrict $\mibar$ to be in $\Theta_0$, as in Lemma
\ref{lem:Rissanen_based}).
Then, for Lebesgue almost all $\mstar \in \Theta_0$:
$$
\liminf_{n \to \infty} \frac{\sum_{i=0}^{n-1}
E_{P_{\mstar}}[D(M_{\mstar}\|M_\mibar)]}{\ln n} \geq \frac{1}{2}
\frac{\underline{V}_\gen}{\overline{V}_\model},
$$
for $\underline{V}_{\gen} := \inf_{\mu \in \Theta_0}
\var_{P_\mu}X$ and $\overline{V}_{\model} := \sup_{\mu \in \Theta_0} \var_{M_\mu}X$.
\end{lemma}
\begin{proof}
Let us denote $\Theta_0 = [\mu_0,\mu_1]$. 
We define a truncated sequence of estimators
$(\mibar')$ as follows:

$$
\mibar' = \left\{
\begin{array}{ll}
 \mu_1 & {\rm if } \;\; \mibar \geq \mu_1 \\
 \mibar & {\rm if } \;\; \mu_0 <
 \mibar < \mu_1 \\ \mu_0 & {\rm if } \;\; \mibar \leq \mu_0
\end{array}
\right.,
$$
so that $\mibar' \in \Theta_0$. Note, that $D(M_{\mstar}\|M_{\mibar}) \geq
D(M_{\mstar}\|M_{\mibar'})$, as there exists $\lambda \in [0,1]$ such
that we can express $\mibar' = \lambda \mstar + (1-\lambda) \mibar$ and
$D(M_{\mstar}\|M_{\lambda \mstar + (1-\lambda) \mibar})$ is strictly
decreasing in $\lambda$ \cite{Grunwald07}. Using this fact and expanding 
$D(M_{\mstar}\|M_{\mibar'})$ into Taylor series as in Lemma
\ref{lem:Rissanen_based}, we get:
\begin{eqnarray*}
&&E_{P_{\mstar}}[D(M_{\mstar}\|M_\mibar)] \geq 
E_{P_{\mstar}}[D(M_{\mstar}\|M_\mibar')] \\ 
&=& \frac{1}{2} E_{P_\mstar} \left[
\frac{(\mibar - \mstar)^2}{\var_{M_\mu} X} \right] 
\geq \frac{1}{2} \frac{1}{\overline{V}_{\model}} E_{P_\mstar}
\left[ (\mibar - \mstar)^2 \right]. 
\end{eqnarray*}
Summing over $i=0,\ldots,n-1$ and using Lemma \ref{lem:Rissanen_based} finishes
the proof.
\end{proof}

Before we prove Theorem~\ref{thm:no_plugin}, we further need a simple lemma to rewrite the 
redundancy in a more convenient form: 
\begin{lemma}
Let $U$ and $\model$ be defined as in Theorem~\ref{thm:no_plugin}.
We have:
$$
\redundancy_U(n)=\sum_{i=0}^{n-1}\,E_{P_\mstar}\left[D(\mmstar\parallel
M_\mibar)\right].$$
\label{lem:redundancy}
\end{lemma}
The usefulness of this lemma comes from the fact that the KL
divergence $D(\cdot \| \cdot)$ is defined as an expectation over
$M_{\mu^*}$ rather than $P_{\mu*}$. 
The proof makes use of a standard result in the theory of exponential families
and can be found e.g. in \cite{Grunwald07} (see also related Lemma 1 in
\cite{GrunwaldD05}).

\begin{proof} {\em (of Theorem~\ref{thm:no_plugin})}
Choose any $\mstar \in \Theta$ and span around it a non-degenerate closed
interval $\Theta'_{\mstar} \subset \Thetamean$, so that $\mstar \in {\rm int}
\Theta'_{\mstar}$. Fix some $\epsilon > 0$. It follows from
general properties of exponential families (see, e.g., \cite{BarndorffNielsen78}) that $\var_{M_\mu}X$ and $\var_{P_\mu}X$
are continuous (with respect to $\mu$), therefore if we choose the interval
$\Theta'_{\mstar}$ small enough, we will have
$\frac{\underline{V}_\gen}{\overline{V}_\model} > \frac{\var_{P_\mstar}X}{\var_{M_\mstar}X} - \epsilon$, with 
$\underline{V}_{\gen} := \inf_{\mu \in \Theta'_\mstar} \var_{P_\mu}X$ and
$\overline{V}_{\model} := \sup_{\mu \in \Theta'_\mstar} \var_{M_\mu}X$.
Using Lemma \ref{lemm:variance_ratio} with $\Theta_0 = \Theta'_\mstar$, and
Lemma \ref{lem:redundancy}, we have for Lebesgue almost all $\mu \in
\Theta'_{\mu^*}$.
$$
\liminf_{n \to \infty} \frac{\mathcal{R}_U(n)}{\ln n} \geq \frac{1}{2}
\frac{\underline{V}_\gen}{\overline{V}_\model} > \frac{1}{2} \left(
  \frac{\var_{P_\mstar}X}{\var_{M_\mstar}X} - \epsilon \right).
$$
Note, that w.l.o.g. $\Theta'_\mstar$ can be chosen to have rational ends. The
family of all intervals $\Theta'_\mstar \subset \Thetamean$ with rational ends and rational $\mstar$,
i.e. $\Xi = \{ \Theta'_\mstar = [\mu_0,\mu_1] \mid \mstar, \mu_0, \mu_1 \in
\Thetamean \cap \fractions \}$, is countable and covers $\Thetamean$,
$\bigcup_{\Theta'_\mstar \in \Xi} \Theta'_\mstar =
\Thetamean$. Therefore,
\begin{multline}
\label{eq:snow}
\text{For Lebesgue almost all $\mstar \in \Thetamean$}:\\
\liminf_{n \to \infty} \frac{\mathcal{R}_U(n)}{\ln n} >
\frac{\var_{P_\mstar}X}{\var_{M_\mstar}X} - \epsilon
\end{multline}
Since this holds for every $\epsilon > 0$, this also  means that $\liminf_{n \to
\infty} \frac{\mathcal{R}_U(n)}{\ln n} \geq
\frac{\var_{P_\mstar}X}{\var_{M_\mstar}X}$ for 
Lebesgue almost all $\mstar \in \Thetamean$. To show this, assume the contrary,
that the set $A=\left\{\mstar \colon \liminf_{n \to
\infty} \frac{\mathcal{R}_U(n)}{\ln n} <
\frac{\var_{P_\mstar}X}{\var_{M_\mstar}X}\right\}$ has positive Lebesgue
measure, $L(A) > 0$. Let $\epsilon_1,\epsilon_2,\ldots$ be any sequence
of positive numbers converging to $0$ and let us define $A_i = \left\{\mstar
\colon \liminf_{n \to \infty} \frac{\mathcal{R}_U(n)}{\ln n} <
\frac{\var_{P_\mstar}X}{\var_{M_\mstar}X} - \epsilon_i \right\}$. Obviously,
$A_1 \subset A_2 \subset \ldots$, and $\bigcup_{i} A_i = A$. From continuity of
measure, we must have $L(A_i) > 0$ for $i$ large enough, which is a
contradiction with (\ref{eq:snow}). The theorem is proved.
\end{proof}

\section*{Proof of Theorem~\ref{thm:robustML}}
We will make use of the following two
theorems, proofs of which can be found in \cite{GrunwaldD05}.
%
%
\begin{theorem}\label{thm:devavasym}
  Let $X, X_1, \ldots$ be i.i.d., let
  $\mnhat:=(n_0\cdot x_0+\sum_{i=1}^n X_i)/(n+n_0)$ and $\mstar=E[X]$.
  If the first $k$ moments of $X$ exist, then
  $E[(\mnhat-\mstar)^k]=O(n^{-\ceil{k\over2}})$.
\end{theorem}

\begin{theorem}\label{thm:boundprob}
Let $X, X_1, \ldots$ be i.i.d. random variables, define
  $\mnhat:=(n_0\cdot x_0+\sum_{i=1}^n X_i)/(n+n_0)$ and
  $\mstar=E[X]$. Let $k\in\{0,2,4,\ldots\}$. If the first $k$ moments
  exists then
  $P(|\mnhat-\mstar|\ge\delta)=O\left(n^{-\ceil{k\over2}}\delta^{-k}\right)$.
\end{theorem}

Before we prove the main theorem, we need the following lemma:
\begin{lemma}
Fix any $s \in \{0,2,4\}$. Let $f(\mu)$ be some continuous function of $\mu$. Suppose
it holds for both $T:=X$ and $T:=-X$ that:
\begin{itemize}
 \item If $T$ is unbounded from above then there is a $k \in \{4,6,\ldots\}$ such that the first $k$ moments of $T$ exist under $P$ and that $f(\mu) = O(\mu^{k-s-2})$.
\item If $T$ is bounded from above by a constant $g$ then $f(\mu)$ is polynomial in $1 / (g - \mu)$.
\end{itemize}
Then the expression $E_P[f(\mu) (\mihat - \mstar)^s]$, for $\mu$ between
$\mstar$ and $\mihat$, is of order $O(i^{-s/2})$.
\label{lemm:bounding}
\end{lemma}
\begin{proof}
The proof follows very closely part of the proof of Lemma 2 in
\cite{GrunwaldD05}; we nevertheless give here a complete
proof for the sake of clarity.

Let us denote $\di := \mihat - \mstar$. We distinguish a number of regions in
the value space of $\di$: let $\Delta_-=(-\infty,0)$ and let $\Delta_0=[0,a)$
for some constant value $a>0$. If the individual outcomes $X$ are bounded on the
right hand side by a value $g$ then we require that $a<g$ and we
define $\Delta_1=[a,g)$; otherwise we define $\Delta_j=[a+j-1, a+j)$
for $j\ge1$. Now we want to analyze asymptotic behavior of:
$$
E_P \left[f(\mu) \di^s \right] 
= \sum_j P(\di\in\Delta_j) E_P \left[f(\mu) \di^s\mid\di\in\Delta_j\right].
$$
If we can establish the proper asymptotic behavior $O(i^{-s/2})$ for all
regions $\Delta_j$ for $j\ge0$, then we can use a symmetrical argument to
establish the behavior for $\Delta_-$ as well, so it suffices if we restrict
ourselves to $j\ge0$. First we show it for $\Delta_0$. In
this case, the basic idea is that since the remainder $f(\mu)$ is
well-defined over the interval $\mstar\le\mu<\mstar+a$, we can bound
it by its extremum on that interval, namely
$m:=\sup_{\mu\in[\mstar,\mstar+a)}\left|f(\mu)\right|$. Now we get:
$$
  \left|P(\di\in\Delta_0)E\left[f(\mu)\di^s
  \mid\di\in\Delta_0\right]\right| ~\le~ 1\cdot
  E\left[\di^s\left|f(\mu)\right|\right],
$$ which is less or equal than $mE\left[\di^s\right]$.
  Using Theorem~\ref{thm:devavasym} we
find that $E[\di^s]$ is $O(i^{-s/2})$, which is what we want.
Theorem~\ref{thm:devavasym} requires that the first four moments of
$P$ exist, but this is guaranteed to be the case: either the outcomes
are bounded from both sides, in which case all moments necessarily
exist, or the existence of the required moments is part of the
condition on the main theorem.

Now we distinguish between the unbounded and bounded
cases. First we assume  $X$ is unbounded from above. In
this case, we must show, hat:
\begin{equation}
\label{eq:unbounded_case}
  \sum_{j=1}^\infty
  P(\di\in\Delta_j)E\left[f(\mu) \di^s\mid\di\in\Delta_j\right] = O(i^{-s/2})
\end{equation}
We bound this expression from above. The $\di$ in the expectation is
at most $a+j$. Furthermore $f(\mu)=O(\mu^{k-s-2})$ by assumption, 
where $\mu\in[a+j-1,a+j)$. Depending on $k$ and $s$, both
boundaries could maximize this function, but it is easy to check that
in both cases the resulting function is $O(j^{k-s-2})$. So we bound
(\ref{eq:unbounded_case}) from the above by:
$$
\sum_{j=1}^\infty P(\left|\di\right|\ge
  a+j-1)(a+j)^	s O(j^{k-s-2}).
$$
Since we know from the condition on the main theorem that the first
$k\ge4$ moments exist, we can apply Theorem~\ref{thm:boundprob} 
to find that $P(|\di|\ge
a+j-1)=O(i^{-\ceil{k\over2}}(a+j-1)^{-k})=O(i^{-{k\over2}})O(j^{-k})$
(since $k$ has to be even); plugging this into the equation and
simplifying we obtain $ O(i^{-{k\over2}})\sum_j O(j^{-2})$, which is of order
$O(i^{-s/2})$, since the sum $\sum_j O(j^{-2})$ converges and $k \geq s$. 

Now we consider the case where the outcomes are bounded from above by
$g$.  This case is more complicated, since now we have made no extra
assumptions as to existence of the moments of $P$. Of course, if the
outcomes are bounded from both sides, then all moments necessarily
exist, but if the outcomes are unbounded from below this may not be
true. To remedy this, we map all outcomes into a new
domain in such a way that all moments of the transformed variables are
guaranteed to exist. Any constant $x^-$ defines a mapping
$g(x):=\max\{x^-,x\}$. We define the random variables
$Y_i:=g(X_i)$, the initial outcome $y_0:=g(x_0)$ and the mapped
analogues of  $\mstar$ and $\mihat$, respectively: $\mdagger$ is defined
as the mean of $Y$ under $P$ and $\mitilde:=(y_0\cdot n_0+\sum_{j=1}^i
Y_j)/(i+n_0)$. Since $\mitilde\ge\mihat$, we can bound:
\begin{eqnarray*}
P(\di\in\Delta_1)\hskip-25pt&&\left|  E\left[f(\mu) \di^s
\mid\di\in\Delta_1\right] \right| \\ 
&\le& P(\mihat-\mstar\ge a) \sup_{\di\in\Delta_1}\left|f(\mu) \di^s \right| \\
  &\le& P(|\mitilde-\mdagger|\ge
  a+\mstar-\mdagger)g^s\sup_{\di\in\Delta_1}\left|f(\mu)\right|%
\end{eqnarray*}
By choosing $x^-$ small enough, we can bring $\mdagger$ and $\mstar$
arbitrarily close together; in particular we can choose $x^-$ such
that $a+\mstar-\mdagger>0$ so that application of
Theorem~\ref{thm:boundprob} is safe. It reveals that the summed
probability is $O(i^{-{k\over2}})$ for any even $k\in\nats$. Now we
bound $f(\mu)$ which is $O((g-\mu)^{-m})$ for some
$m\in\nats$ by the condition on the main theorem. Here we use that
$\mu\le\mihat$; the latter is maximized if all outcomes equal the
bound $g$, in which case the estimator equals
$g-n_0(g-x_0)/(i+n_0)=g-O(i^{-1})$. Putting all of this together, we
get $\sup\left|f(\mu)\right|=O((g-\mu)^{-m})=O(i^m)$; if we
plug this into the equation we obtain:
\begin{displaymath}
  \ldots~\le~\sum_i O(i^{-{k\over2}})g^sO(i^m)=g^s\sum_i O(i^{m-{k\over2}})
\end{displaymath}
This is of order $O(i^{-s/2})$ if we choose $k\ge6m+s$. We can do this because
the construction of $g(\cdot)$ ensures that all moments exist,
and therefore certainly the first $6m+s$ moments.
\end{proof}

We can now proceed to prove the theorem:

\begin{proof}{\em (of Theorem \ref{thm:robustML})}
We  express the relative redundancy of the slightly squashed
ML plug-in code $U$ by the sum of the relative redundancy of the ordinary
ML plug-in code $\Uhat$ and the difference in expected codelengths between
$U$ and $\Uhat$:
\begin{eqnarray*}
&&\redundancy_{U}(n) = E_{P}[L_U(Z^n)] - 
E_{P}[- \ln \mmstar(Z^n)]. \\
&=& E_{P}[L_U(Z^n) - L_{\Uhat}(Z^n)] + \redundancy_{\Uhat}(n)
\\ &=& E_{P}[L_U(Z^n) - L_{\Uhat}(Z^n)] + \frac{1}{2}
\frac{\var_P X}{\var_{M_{\mu^*}} X} \ln n + O(1),
\end{eqnarray*}
where the last equality follows from (\ref{eq:redundancyML}), which is valid
under the conditions imposed on ${d^4 \over
  d\mu^4} D(M_\mstar \| \mmu)$ (see Condition 1 in \cite{GrunwaldD05} for
  details). We have:
$$
\begin{array}{ll}
&L_U(Z^n) - L_{\Uhat}(Z^n) \\
=& \sum_{i=0}^{n-1} \left(- \ln U(Z_{i+1} \mid Z_i) +
\ln {\Uhat}(Z_{i+1} \mid Z_i) \right)  \\
=& \sum_{i=0}^{n-1} \left( \ln\left(1 + \frac{1}{2i}\right) -
\ln\left(1 + \frac{1}{2i} I_\model(\mihat)(X_{i+1} -
\mihat)^2\right)\right) .
\end{array}
$$
Since $\ln\left(1+\frac{1}{2i}\right) = \frac{1}{2i} + O(i^{-2})$, we have:
\begin{equation}
\label{eq:logarithm}
\sum_{i=0}^{n-1} \ln\left(1 + \frac{1}{2i}\right) = \frac{1}{2} \ln n + O(1).
\end{equation}
To analyze the second term in the sum, we use the fact that for
arbitary $a\geq 0$:
$$
 -a \leq -\ln(1+a) \leq -a  + \frac{1}{2} a^2,
$$
which follows e.g. from expanding the logarithm into Taylor
expansion up to the second order. In our case, $a = V_i :=\frac{1}{2i}
I_\model(\mihat)(X_{i+1} - \mihat)^2$. We will show that $E_P[V_i^2]$ is
$O(i^{-2})$, and then $E_P[-\ln(1+V_i)] = -E_P[V_i] + O(i^{-2})$.
We have:
$$
\begin{array}{l@{~}l}
&E_P \left[ V_i^2 \right] = \frac{1}{4i^2} E_P \left[
I_\model^2(\mihat) \left( X_{i+1} - \mihat\right)^4 
\right] \\ 
=&\frac{1}{4i^2}  E_{X^i \sim P} \left[ I_\model^2 (\mihat)E_{X_{i+1}
\sim P} \left[\left( X_{i+1} - \mstar + \mstar - \mihat\right)^4\right] \right] \\
=&\frac{1}{4i^2}  E_P \left[ I_\model^2 (\mihat) \left( \moment{4} -
4\delta_i\moment{3}  + 6  \delta_i^2\var_{P_\mstar}X + \delta_i^4 \right)
\right],
\end{array}
$$
where $\moment{k}$ is $E_P[(X-\mstar)^k]$, the $k$-th central moment of
$P_{\mstar}$, and $\delta_i = \mihat - \mstar$. We will show that the terms
under expectation are bounded. If $I_\model(\mihat)$ is constant,
then we apply  Theorem \ref{thm:devavasym} with $k=1$, $k=2$ and $k=4$ to the
second, third and fourth term, respectively and thus all the terms under
expectation are $O(1)$. If $I_\model(\mihat)$ is not constant, then by Condition 
\ref{condition} the assumptions of Lemma
\ref{lemm:bounding} are satisfied with $f(\mu) = I^2_\model(\mu)$ and 
$s=0,2,4$. Applying the lemma subsequently to the first, third and fourth term
(with $s=0,2,4$, respectively), we see that all those terms are $O(1)$. The
second term is also $O(1)$ by applying Lemma \ref{lemm:bounding} once again
with $f(\mu) = \mu I^2_\model(\mu)$ and $s=0$ (assumptions are again
satisfied by Condition \ref{condition}). Thus, we showed that $E_P[V_i^2] =
O(i^{-2})$.

Next, we consider $E_P[V_i]$:
\begin{eqnarray*}
&&E_P[V_i] = \frac{1}{2i} E_P \left[ 
I_\model(\mihat)(X_{i+1} - \mstar + \mstar - \mihat)^2 \right] \\
&=& \frac{1}{2i} E_P \left[I_\model(\mihat)\left( \var_{P_\mstar}X + \delta_i^2
 \right) \right] \\
&=& \frac{1}{2i} \left(\var_{P_\mstar}X E_P \left[I_\model(\mihat)\right] + E_P
\left[I_\model(\mihat) \delta_i^2\right] \right).
\end{eqnarray*}
The second term $E_P\left[I_\model(\mihat) \delta_i^2\right]$ is $O(i^{-1})$ by
Lemma \ref{lemm:bounding} applied with $f(\mu) = I_\model(\mu)$ and $s=2$. To
analyze the first term we expand $I_\model(\mihat)$ into Taylor series around
$\mstar$: 
$$
E_P\left[I_\model(\mihat)\right] = I_\model(\mstar) + E_P\left[\frac{d}{d\mu}
I_\model(\mstar) \delta_i + \frac{d^2}{d^2\mu} I_\model(\mu)
\delta^2_i\right], $$
for some $\mu$ between $\mstar$ and $\mihat$. The linear
term in the expansion is $O(i^{-1})$ by Theorem \ref{thm:devavasym} applied
with $k=1$. The quadratic term is $O(i^{-1})$ by applying
Lemma \ref{lemm:bounding} with $f(\mu) = \frac{d^2}{d^2\mu} I_\model(\mu)$ and $s=2$;
Condition \ref{condition} guarantees that assumptions of the lemma
are satisfied. Thus, using (\ref{eq:efficient}):
$$
\!E_P[V_i] \!=\! \frac{1}{2i} I_{\model}(\mstar) \var_{P_\mstar}\!X + O(i^{-2}) \!=\!
\frac{1}{2i} \frac{\var_{P_\mstar}\!X}{\var_{M_\mstar}\!X} + O(i^{-2}), 
$$
so that:
\begin{equation}
\label{eq:ln_V}
E_P[-\ln(1+V_i)] = - \frac{1}{2i} \frac{\var_{P_\mstar}X}{\var_{M_\mstar}X} + O(i^{-2}), 
\end{equation}
Taking together (\ref{eq:ln_V}) and (\ref{eq:logarithm}) we have:
$$
L_U(Z^n) - L_{\Uhat}(Z^n) = \frac{1}{2} \ln n - \frac{1}{2}
\frac{\var_{P_\mstar}X}{\var_{M_\mstar}X} \ln n+ O(1),
$$
and thus:
$$
R_U(n) = \frac{1}{2} \ln n + O(1).
$$
 \end{proof}

\end{document}